\documentclass[11pt]{article}
\usepackage[margin=1in]{geometry}
\usepackage[colorlinks=true, allcolors=blue]{hyperref}

\usepackage{fixmath}
\usepackage{bm}
\usepackage{amsbsy}
\usepackage{color}
\usepackage{verbatim}
\usepackage{multirow}
\usepackage{amssymb}
\usepackage{amsthm}
\usepackage{array}
\usepackage{mathtools}
\usepackage{amsmath}
\usepackage{graphicx}
\usepackage{tikz}
\usetikzlibrary{positioning}
\usepackage{xcolor}
\usepackage{thm-restate}

\newtheorem{theorem}{Theorem}
\newtheorem{lemma}[theorem]{Lemma}

\newtheorem{assumption}[theorem]{Assumption}

\newcommand{\ceil}[1]{\ensuremath{\lceil#1\rceil}}

\newcommand{\Ceil}[1]{\ensuremath{\left\lceil#1\right\rceil}}

\newcommand{\lrA}[1]{\ensuremath{\left(#1\right)}}
\newcommand{\lrB}[1]{\ensuremath{\left[#1\right]}}

\def\OPT{\mbox{OPT}}

\newcommand{\EE}[1]{\ensuremath{\mathbb{E}[#1]}}
\newcommand{\EEE}[1]{\ensuremath{\mathbb{E}\lrB{#1}}}

\title
{
A 5-approximation Algorithm for the Traveling Tournament Problem
}

\author
{
Jingyang Zhao\footnote{University of Electronic Science and Technology of China. Email: \texttt{jingyangzhao1020@gmail.com}.}
\and
Mingyu Xiao\footnote{University of Electronic Science and Technology of China. Email: \texttt{myxiao@gmail.com}.}
}

\date{}

\begin{document}

\maketitle

\begin{abstract}
The Traveling Tournament Problem (TTP-$k$) is a well-known benchmark problem in tournament timetabling, which asks us to design a double round-robin schedule such that the total traveling distance of all $n$ teams is minimized under the constraints that each pair of teams plays one game in each other's home venue, and each team plays at most $k$-consecutive home games or away games. Westphal and Noparlik (Ann. Oper. Res. 218(1):347-360, 2014) claimed a $5.875$-approximation algorithm for all $k\geq 4$ and $n\geq 6$.
However, there were both flaws in the construction of the schedule and in the analysis.
In this paper, we show that there is a 5-approximation algorithm for all $k$ and $n$. Furthermore, if $k \geq n/2$, the approximation ratio can be improved to 4.

\end{abstract}

\section{Introduction}
In the field of sports scheduling~\cite{kendall2010scheduling}, the traveling tournament problem (TTP-$k$) is a well-known benchmark problem that was first systematically introduced in~\cite{easton2001traveling}. This problem aims to find a double round-robin tournament satisfying some constraints, minimizing the total distance traveled by all participant teams. In a double round-robin tournament of $n$ teams, each team will play 2 games against each of the other $n-1$ teams, one home game at its home venue and one away game at its opponent's home venue. Additionally, each team should play one game a day, all games need to be scheduled on $2(n-1)$ consecutive days, and so there are exactly $n/2$ games on each day. According to the definition, we know that $n$ is always even. For TTP-$k$, we have the following three basic constraints or assumptions on the double round-robin tournament.

\medskip
\noindent
\textbf{Traveling Tournament Problem (TTP-$k$)}
\begin{itemize}
\item \emph{No-repeat}: Two teams cannot play against each other on two consecutive days.
\item \emph{Direct-traveling}: Before the first game starts, all teams are at home, and they will return home after the last game ends. Furthermore, a team travels directly from its game venue on the $i$-th day to its game venue on the $(i+1)$-th day.
\item \emph{Bounded-by-$k$}: Each team can have at most $k$-consecutive home games or away games.
\end{itemize}

The smaller the value of $k$, the more frequently a team has to return home. By contrast, if $k$ is very large, say $k=n-1$, then the bounded-by-$k$ constraint loses meaning and a team can schedule their travel distance as short as that in the traveling salesman problem (TSP).

The input of TTP-$k$ is a complete graph where each vertex represents a team and the distance between two vertices $i$ and $j$, denoted by $d(i,j)$, is the distance from the home of team $i$ to the home of team $j$. In this paper, we only consider the case that the distance function $d$ satisfies the symmetry and triangle inequality properties, i.e., $d(i,j) = d(j,i)$ and $d(i,h) \leq d(i,j) + d(j,h)$ for all $1\leq i,j,h\leq n$.

TTP-$k$ is a difficult optimization problem. Readers can refer to \cite{DBLP:journals/eor/BulckGSG20,duran2021sports} for an overview. For the case of $k=1$, the problem is infeasible~\cite{de1988some}. For $3\leq k=O(1)$ or $k=n-1$, the NP-hardness has been established~\cite{bhattacharyya2016complexity, thielen2011complexity, DBLP:journals/corr/abs-2110-02300}. Recently, Bendayan \emph{et al.}~\cite{2022apx} further proved the APX-hardness for $k=n-1$. Although the hardness of TTP-2 is still not formally proved, it is believed that TTP-2 is also hard since it is not easy to construct a feasible solution to it. In the literature, there is a large number of contributions on approximation algorithms~\cite{thielen2012approximation,DBLP:conf/mfcs/XiaoK16,DBLP:conf/ijcai/ZhaoX21,DBLP:conf/cocoon/ZhaoX21,DBLP:conf/atmos/ChatterjeeR21,imahori20211+,miyashiro2012approximation,2012LDTTP,hoshino2013approximation} and heuristic algorithms~\cite{easton2003solving,lim2006simulated,anagnostopoulos2006simulated,di2007composite,goerigk2014solving,DBLP:journals/anor/GoerigkW16}.

For heuristic algorithms, most known works are concerned with the case of $k=3$. Since the search space is usually very large, many instances of TTP-3 with more than 10 teams in the online benchmark \cite{trick2007challenge,DBLP:journals/eor/BulckGSG20} have not been completely solved even by using high-performance machines.

In terms of approximation algorithms, almost all results are based on the assumption that the distance holds the symmetry and triangle inequality properties. This is natural and practical in the sports schedule. For $k=2$, the approximation ratios for even $n/2$ and odd $n/2$ have been improved to $(1+3/n)$ and $(1+5/n)$ \cite{2023ttp2}. For $k=3$ and $k=4$, the best approximation ratios are $(139/87+\varepsilon)$ and $(17/10+\varepsilon)$, respectively~\cite{zhao2022improved}. For $5\leq k=o(n)$, the ratio is $(5k-7)/(2k)+O(k/n)$~\cite{yamaguchi2009improved}. For $k=n-1$, Imahori~\emph{et al.}~\cite{imahori2010approximation} proposed a $2.75$-approximation algorithm. At the same time, Westphal and Noparlik~\cite{westphal2014} claimed a $5.875$-approximation algorithm for all $k\geq 4$ and $n\geq 6$. However, there were two flaws in the algorithm and analysis.
The first one is in the construction of the schedule which can lead to the output of an infeasible solution. The second one is in the analysis of the total cost which can lead to a worse approximation ratio.
To explain the flaws clearly, we give the details in the appendix.

In this paper, we will show that there is
a 5-approximation algorithm for TTP-$k$ with all $k$ and $n$. The approximation ratio can be further improved to 4 if $k \geq n/2$. Our algorithm uses a similar construction in~\cite{westphal2014}.
We will avoid the flaw and introduce the algorithm in a randomized way, which will simplify our analysis. For the analysis, we refined some lower bounds and analyze the solution quality in a different way by using the refined bounds.

Our algorithm will be introduced in Section~\ref{sec_2}.
In Section~\ref{mainalgorithm}, we show the framework of our approximation algorithm, and the randomzied algorithm to generate our parameters.
In Section~\ref{construction}, we explain the main construction algorithm in detail and prove its feasibility.
In Section~\ref{analysis}, we analyze the approximation quality of our algorithm. Specifically, in Section~\ref{bounds} we propose some useful bounds, in Section \ref{ouranalysis} we do some analysis, and in Section \ref{ourratio}, we improve the approximation ratio. Section \ref{conclusion} makes some concluding remarks.
We also show the flaw in the previous construction in Appendix~\ref{previousconstruction} and the flaw in the previous analysis in Appendix~\ref{previousanalysis}.

\section{The Algorithm}\label{sec_2}
Let $n$ denote the number of teams, where $n$ is even. The set of $n$ teams is denoted by $\{t_1,t_2,\dots, t_n\}$. We use $G=(V, E)$ to denote the complete graph on $n$ vertices $\{1,2,\dots,n\}$ representing the $n$ teams. There is a distance/length function $d: E\to \mathbb{R}_{\geq0}$ on the edges of $G$. The distance of edge $ij$, denoted by $d(i,j)$, is the distance between the homes of teams $t_i$ and $t_j$. We also let $s(i)=\sum_{j\neq i}d(i,j)$, i.e., the total distance of edges incident on vertex $i$ in $G$, and let $\Delta=\sum_is(i)$. Given any (Hamiltonian) cycle $T$ of graph $G$, we use $d(T)$ to denote the length of $T$, i.e., the total length of edges on $T$. We use $T^*$ to denote an optimal Hamiltonian cycle, i.e., the Hamiltonian cycle with a minimum length.
In this paper, we consider that $k$ is a part of the input. Recall that TTP-1 is infeasible~\cite{de1988some}, and for TTP-$k$ with $k\geq n$, the problem is equivalent with TTP-$(n-1)$. Hence, we assume w.l.o.g. that $2\leq k<n$.

\subsection{The Algorithm Framework}\label{mainalgorithm}
Assume that we are given a Hamiltonian cycle $T$ of graph $G$, which can be computed by a polynomial-time approximation algorithm. Let $\sigma: V\leftrightarrow\{1,\dots,n\}$ be a permutation of the $n$ teams, i.e., a bijection label function that maps the $n$ teams. Specifically, we assume that the team represented by the number $i$ is labeled as $t_i$. Let $l\in \{1,2,\dots,k\}$
be a parameter. Given $\sigma$ and $l$, we will construction a feasible schedule for the teams. The construction algorithm is denoted by  CONS($\sigma$,$l$). We delay the description of CONS($\sigma$,$l$) to the next subsection.

For every different $\sigma$ and $l$ we can get a feasible solution. To get a solution, we want to find good parameters $\sigma$ and $l$.
We will use a simple randomized algorithm to generate the two parameters. The algorithm can be easily derandomized in polynomial time and it will simplify the analysis.

The randomized algorithm contains the following three steps.

\medskip

\noindent\textbf{Step~1.} Select one vertex from $T$ such that the total distance of edges incident on it is minimized, and label it as $n$.

\noindent\textbf{Step~2.} Let $T'$ be the cycle obtained by shortcutting the vertex $n$ from $T$. W.l.o.g., orient $T'$ with an arbitrary direction. To label the vertices of $T'$, we choose one vertex of $T'$ uniformly at random, label it as $1$, and then label the vertices $2,3,\dots, n-1$ following the orientation of $T'$, respectively. Note that there are $n-1$ cases to label the oriented cycle $T'$.

\noindent\textbf{Step~3.} The team represented by the vertex $i$ is labeled as $t_i$.
\begin{itemize}
    \item If $k<n/2$, the parameter $l$ is taken from $\{1,2,\dots,k\}$ uniformly at random;
    \item Otherwise, let $l=n/2-1$.
\end{itemize}

The randomized algorithm will generate two randomized parameters $\sigma$ and $l$. In order to simplify the argument and analysis, we first derandomize the parameter $l$ as follows:
If $k<n/2$, compute CONS($\sigma$,$l$) for each $l\in\{1,2,\dots,k\}$, and return the best one from the $k$ schedules. The derandomization only increases a factor of $k$ in the running time.
Next, we assume that $l$ is a fixed parameter chosen as above. When we talk about expectation, it means the expectation with respect to $\sigma$.

\subsection{The Construction Algorithm}\label{construction}
In this subsection, we introduce the construction algorithm CONS($\sigma$,$l$). Given fixed $\sigma$ and $l$, we describe the construction.
Recall that the schedule contains $2(n-1)$ days, and each day contains $m\coloneqq n/2$ games. It can be split into two \emph{seasons}, where the first $n-1$ days are seen as the first season, and the other $n-1$ days are seen as the second season.

\subsubsection{The framework of the construction}
We first show the games in the first season.

On the first day of the schedule, the games are shown in Figure~\ref{fig01}. The last team $t_n$ is represented by a double-cycle node, while each of the other $n-1$ teams is represented by a single-cycle.
Within this figure, there are $m$ directed edges between teams, where a directed edge from team $t_{i'}$ to team $t_{i''}$ means a game between these two teams, with the game taking place at the home venue of team $t_{i''}$. Thus, the $n$ teams form $m$ games.
For the sake of presentation, the setting of the directions of these edges is explained later. We can currently ignore the boxes in the figure now.

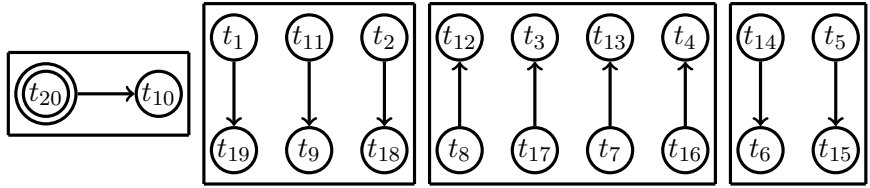
\begin{figure}[ht]
\centering
\begin{tikzpicture}
[
leftsuperteam/.style={circle, draw=black!100, very thick, minimum size=6mm, inner sep=0pt},
normalsuperteam/.style={circle, draw=black!100, very thick, minimum size=6mm, inner sep=0pt},
rightsuperteam/.style={draw=black!100, very thick, minimum size=6mm, inner sep=0pt},
cyc/.style={circle, draw=black!100, very thick, minimum size=8mm, inner sep=0pt},
]
\draw[very thick] (1.4,0.55) to (1.4,-0.55);
\draw[very thick] (-1,0.55) to (-1,-0.55);
\draw[very thick] (-1,0.55) to (1.4,0.55);
\draw[very thick] (-1,-0.55) to (1.4,-0.55);
\draw[very thick] (1.6,1.2) to (4.4,1.2);
\draw[very thick] (1.6,-1.2) to (4.4,-1.2);
\draw[very thick] (1.6,1.2) to (1.6,-1.2);
\draw[very thick] (4.4,1.2) to (4.4,-1.2);
\draw[very thick] (4.6,1.2) to (8.4,1.2);
\draw[very thick] (4.6,-1.2) to (8.4,-1.2);
\draw[very thick] (4.6,1.2) to (4.6,-1.2);
\draw[very thick] (8.4,1.2) to (8.4,-1.2);
\draw[very thick] (8.6,1.2) to (10.4,1.2);
\draw[very thick] (8.6,-1.2) to (10.4,-1.2);
\draw[very thick] (8.6,1.2) to (8.6,-1.2);
\draw[very thick] (10.4,1.2) to (10.4,-1.2);
\node[leftsuperteam]        at (-0.5,0) {$t_{20}$};
\node[cyc]        (l)   at (-0.5,0) {};
\node[normalsuperteam]      (r)   at (1,0) {$t_{10}$};
\draw[very thick,->,above] (l.east) to (r.west);
\foreach \i in {2}
{
\node[normalsuperteam]      (u)    at (\i, 0.75) {$t_{1}$};
\node[normalsuperteam]      (d)    at (\i, -0.75) {$t_{19}$};
\draw[very thick,<-] (d.north) to (u.south);
}
\foreach \i in {3}
{
\node[normalsuperteam]      (u)    at (\i, 0.75) {$t_{11}$};
\node[normalsuperteam]      (d)    at (\i, -0.75) {$t_9$};
\draw[very thick,<-] (d.north) to (u.south);
}
\foreach \i in {4}
{
\node[normalsuperteam]      (u)    at (\i, 0.75) {$t_{2}$};
\node[normalsuperteam]      (d)    at (\i, -0.75) {$t_{18}$};
\draw[very thick,<-] (d.north) to (u.south);
}
\foreach \i in {5}
{
\node[normalsuperteam]      (u)    at (\i, 0.75) {$t_{12}$};
\node[normalsuperteam]      (d)    at (\i, -0.75) {$t_8$};
\draw[very thick,<-] (u.south) to (d.north);
}
\foreach \i in {6}
{
\node[normalsuperteam]      (u)    at (\i, 0.75) {$t_{3}$};
\node[normalsuperteam]      (d)    at (\i, -0.75) {$t_{17}$};
\draw[very thick,<-] (u.south) to (d.north);
}
\foreach \i in {7}
{
\node[normalsuperteam]      (u)    at (\i, 0.75) {$t_{13}$};
\node[normalsuperteam]      (d)    at (\i, -0.75) {$t_7$};
\draw[very thick,<-] (u.south) to (d.north);
}
\foreach \i in {8}
{
\node[normalsuperteam]      (u)    at (\i, 0.75) {$t_{4}$};
\node[normalsuperteam]      (d)    at (\i, -0.75) {$t_{16}$};
\draw[very thick,<-] (u.south) to (d.north);
}
\foreach \i in {9}
{
\node[normalsuperteam]      (u)    at (\i, 0.75) {$t_{14}$};
\node[normalsuperteam]      (d)    at (\i, -0.75) {$t_6$};
\draw[very thick,<-] (d.north) to (u.south);
}
\foreach \i in {10}
{
\node[normalsuperteam]      (u)    at (\i, 0.75) {$t_{5}$};
\node[normalsuperteam]      (d)    at (\i, -0.75) {$t_{15}$};
\draw[very thick,<-] (d.north) to (u.south);
}
\end{tikzpicture}
\caption{The schedule on the first day, where $n=20$, $k=4$, $l=2$, and $b=4$}
\label{fig01}
\end{figure}

On the second day of the schedule, the position of the double-cycle node representing team $t_n$ remains unchanged. However, the positions of the $n-1$ single-cycle nodes representing the other teams are altered by shifting them one position in the counterclockwise direction.
An illustration of the schedule on the second day is shown in Figure~\ref{fig02}.

\begin{figure}[ht]
\centering
\begin{tikzpicture}
[
leftsuperteam/.style={circle, draw=black!100, very thick, minimum size=6mm, inner sep=0pt},
normalsuperteam/.style={circle, draw=black!100, very thick, minimum size=6mm, inner sep=0pt},
rightsuperteam/.style={draw=black!100, very thick, minimum size=6mm, inner sep=0pt},
cyc/.style={circle, draw=black!100, very thick, minimum size=8mm, inner sep=0pt},
]
\draw[very thick] (1.4,0.55) to (1.4,-0.55);
\draw[very thick] (-1,0.55) to (-1,-0.55);
\draw[very thick] (-1,0.55) to (1.4,0.55);
\draw[very thick] (-1,-0.55) to (1.4,-0.55);
\draw[very thick] (1.6,1.2) to (4.4,1.2);
\draw[very thick] (1.6,-1.2) to (4.4,-1.2);
\draw[very thick] (1.6,1.2) to (1.6,-1.2);
\draw[very thick] (4.4,1.2) to (4.4,-1.2);
\draw[very thick] (4.6,1.2) to (8.4,1.2);
\draw[very thick] (4.6,-1.2) to (8.4,-1.2);
\draw[very thick] (4.6,1.2) to (4.6,-1.2);
\draw[very thick] (8.4,1.2) to (8.4,-1.2);
\draw[very thick] (8.6,1.2) to (10.4,1.2);
\draw[very thick] (8.6,-1.2) to (10.4,-1.2);
\draw[very thick] (8.6,1.2) to (8.6,-1.2);
\draw[very thick] (10.4,1.2) to (10.4,-1.2);
\node[leftsuperteam]        at (-0.5,0) {$t_{20}$};
\node[cyc]        (l)   at (-0.5,0) {};
\node[normalsuperteam]      (r)   at (1,0) {$t_{1}$};
\draw[very thick,->,above] (l.east) to (r.west);

\foreach \i in {2}
{
\node[normalsuperteam]      (u)    at (\i, 0.75) {$t_{11}$};
\node[normalsuperteam]      (d)    at (\i, -0.75) {$t_{10}$};
\draw[very thick,<-] (d.north) to (u.south);
}
\foreach \i in {3}
{
\node[normalsuperteam]      (u)    at (\i, 0.75) {$t_{2}$};
\node[normalsuperteam]      (d)    at (\i, -0.75) {$t_{19}$};
\draw[very thick,<-] (d.north) to (u.south);
}
\foreach \i in {4}
{
\node[normalsuperteam]      (u)    at (\i, 0.75) {$t_{12}$};
\node[normalsuperteam]      (d)    at (\i, -0.75) {$t_{9}$};
\draw[very thick,<-] (d.north) to (u.south);
}
\foreach \i in {5}
{
\node[normalsuperteam]      (u)    at (\i, 0.75) {$t_{3}$};
\node[normalsuperteam]      (d)    at (\i, -0.75) {$t_{18}$};
\draw[very thick,<-] (u.south) to (d.north);
}
\foreach \i in {6}
{
\node[normalsuperteam]      (u)    at (\i, 0.75) {$t_{13}$};
\node[normalsuperteam]      (d)    at (\i, -0.75) {$t_{8}$};
\draw[very thick,<-] (u.south) to (d.north);
}
\foreach \i in {7}
{
\node[normalsuperteam]      (u)    at (\i, 0.75) {$t_{4}$};
\node[normalsuperteam]      (d)    at (\i, -0.75) {$t_{17}$};
\draw[very thick,<-] (u.south) to (d.north);
}
\foreach \i in {8}
{
\node[normalsuperteam]      (u)    at (\i, 0.75) {$t_{14}$};
\node[normalsuperteam]      (d)    at (\i, -0.75) {$t_{7}$};
\draw[very thick,<-] (u.south) to (d.north);
}
\foreach \i in {9}
{
\node[normalsuperteam]      (u)    at (\i, 0.75) {$t_{5}$};
\node[normalsuperteam]      (d)    at (\i, -0.75) {$t_{16}$};
\draw[very thick,<-] (d.north) to (u.south);
}
\foreach \i in {10}
{
\node[normalsuperteam]      (u)    at (\i, 0.75) {$t_{15}$};
\node[normalsuperteam]      (d)    at (\i, -0.75) {$t_{6}$};
\draw[very thick,<-] (d.north) to (u.south);
}
\end{tikzpicture}
\caption{The schedule on the second day, where $n=20$, $k=4$, $l=2$, and $b=4$}
\label{fig02}
\end{figure}
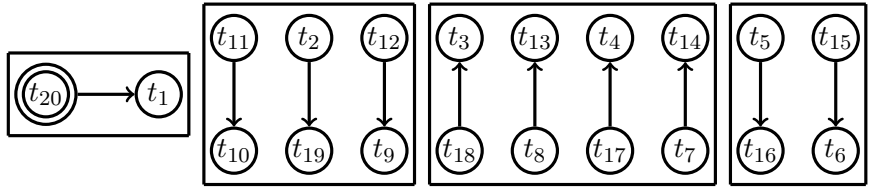

Analogously, we obtain the games on rest days in the first season based on the rotation scheme.

The first season can be presented by $\Gamma_1\cdot \Gamma_2\cdots \Gamma_{n-2}\cdot \Gamma_{n-1}$, where $\Gamma_i$ denotes the $m$ games on the $i$-th day. In the first season, it is easy to see that every pair of teams has played exactly once. Let $\overline{\Gamma_i}$ denote the $m$ games of $\Gamma_i$ but with reversed home venues. Then, there are still $n-1$ days of unarranged games $\{\overline{\Gamma_1}, \overline{\Gamma_2},\dots,\overline{\Gamma_{n-1}}\}$, which will be arranged in the second season in the order of $\overline{\Gamma_{n-2}\cdot \Gamma_{n-1}\cdot \Gamma_1\cdot \Gamma_2\cdots \Gamma_{n-3}}$.
Therefore, the complete schedule can be presented by
\[
\Gamma_1\cdot \Gamma_2\cdots \Gamma_{n-2}\cdot \Gamma_{n-1}\cdot \overline{\Gamma_{n-2}\cdot \Gamma_{n-1}\cdot \Gamma_1\cdot \Gamma_2\cdots \Gamma_{n-3}}.
\]

\subsubsection{The details of the construction}
Recall that $l\in\{1,2,\dots,k\}$ is an input of the algorithm. Let $b\coloneqq\ceil{\frac{m-l-1}{k}}+2$.
To determine the directions of the edges in the schedule, we divide them into $b$ blocks, as shown in Figures~\ref{fig01} and \ref{fig02}.
Each block contains a specific number of edges, known as its width. We denote the width of the $i$-th block, from left to right, as $w_i$.
In our construction, we set $w_1=1$, $w_2=m-1-(b-3)k-l$, $w_i=k$ for $i=3,\dots, b-1$, and $w_b=l$. Note that the sum of the widths of all the blocks equals the total number of edges, i.e., $\sum_{i=1}^{b}w_i=m$.

When $k<n/2$, with $l\in\{1,2,\dots,k\}$, there are $k$ cases to consider. On the other hand, when $k\geq n/2$, we only need to consider the case where $l=n/2-1$. In this specific case, there are only two blocks: the first block with width $w_1=1$ and the second block with width $w_2=m-1$. Additionally, if $(m-1-l)\bmod k=0$, we have $w_2=k$; otherwise, we have $w_2=(m-1-l)\bmod k$.

Note that the directions of the edges in the second season are arranged based on the first season. Hence, we only explain the directions of the edges in the first season. There are two rules.

\textbf{Rule 1.} Initially, the directions of edges in the same block typically are the same, and the direction in each block alternatively changes from the $2$-nd block to the $b$-th block. Moreover, the direction of the edge in the 1-st block is from left to right while the directions of the edges in the 2-nd block are from above to below.
This arrangement ensures that, in Figure~\ref{fig01}, the directions of the edges in the 3-rd block are from below to above, the directions of the edges in the 4-th block are from above to below, and so on.
Since every $i$-th block with $3\leq i\leq b-1$ has a width of $k$, this can make sure that every team plays $k$-consecutive home/away games from the beginning of the entry to the end of the exit within one season. And, it helps prevent any team from playing more than $k$-consecutive home/away games since the direction in each block alternatively changes.
A special case is that
\begin{itemize}
    \item if $w_2=k$, we will reverse the direction of the most left edge in the 2-nd block, i.e., the edge $t_1\rightarrow t_{19}$ in Figure~\ref{fig01}, for all days in the first season.
\end{itemize}
Therefore, there are at most $k-1$ edges with the same direction in the 2-nd block. One can imagine that the 2-nd block is further divided into two blocks in this case.

\textbf{Rule 2.} Then, we make sure that the direction of the most left edge, which is incident on team $t_n$, changes every $k$-th day, while the directions of the remaining $m-1$ edges remain unchanged. This means that for the edge $t_{20}\rightarrow t_{10}$ in Figure~\ref{fig01}, the direction is from left to right on the 1-st day to the $k$-th day, from right to left on the $(k+1)$-th day to the $2k$-th day, and so on.
In the specific case where $k=4$, it is noteworthy that the directions of the edges on the second day, as shown in Figure~\ref{fig02}, are the same as those on the first day depicted in Figure~\ref{fig01}.
Moreover, let $r=(n-1)\bmod 2k$. A special case is that
\begin{itemize}
    \item if $r\leq k<n-1$, we will reverse the most left edge only on the 1-st day.
\end{itemize}
The reason is that we need to ensure that team $t_n$ dose not play more than $k$-consecutive home/away games during the transition between the first and second seasons.

\subsubsection{The feasibility of the construction}
We first show that in the construction all teams will not play more than $k$-consecutive home/away games.

\begin{lemma}\label{feasibility1}
No teams play more than $k$-consecutive home/away games.
\end{lemma}
\begin{proof}
Consider an arbitrary team $t_i$ with $i<n$. Recall that there are at most $k-1$ edges with the same direction in the 2-nd block. The construction is based on the rotation scheme. It is easy to verify that $t_i$ plays at most $k$-consecutive home/away games in the first season. The second season uses the same rotation scheme. Analogously, $t_i$ also plays at most $k$-consecutive home/away games in the second season. If $t_i$ plays more than $k$-consecutive home/away games, then these games must include the last day of the first season and the first day of the second season, i.e., $\Gamma_{n-1}$ and $\overline{\Gamma_{n-2}}$. Therefore, $t_i$ plays both home/away games in these two days. We assume w.l.o.g. that $t_i$ plays home games. Then, it plays away games in $\Gamma_{n-2}$ and $\overline{\Gamma_{n-1}}$. Note that these four days are arranged in the order of $\Gamma_{n-2}\cdot\Gamma_{n-1}\cdot\overline{\Gamma_{n-2}\cdot\Gamma_{n-1}}$. Hence, there are at most $2$-consecutive home games containing $\Gamma_{n-1}$ and $\overline{\Gamma_{n-2}}$, a contradiction.

Next, we consider team $t_n$. Recall that $r=(n-1)\bmod 2k$. We consider the home/away patterns of team $t_n$ in the first and the second seasons. For the sake of presentation, we use `$H^p$' and `$A^p$' to denote $p$-consecutive home games and $p$-consecutive away games, respectively. We consider the following two cases.

\textbf{Case~1: $r>k$.} Recall that the direction of the most left edge (incident on team $t_n$) changes every $k$-th day. The home/away pattern of team $t_n$ in the first season is $A^kH^k\cdots A^kH^kA^kH^{r-k}$. In the second season, if $r-k=1$, the pattern is $HAH^kA^k\cdots H^kA^kH^{k-1}$; otherwise, the pattern is $A^2H^kA^k\cdots H^kA^kH^kA^{r-k-2}$. In both cases, the patterns in these two seasons can be combined well without creating more than $k$-consecutive home/away games.

\textbf{Case~2: $r\leq k$.} Recall that we further reverse the most left edge on the 1-st day if $r\leq k<n-1$. Note that we have $r\geq 1$ since $n$ is even. If $k\geq n-1$, the pattern is obviously feasible. If $k<n-1$, since $r\leq k$, we have $n-1>2k$. Therefore, the home/away pattern of team $t_n$ in the first season is $HA^{k-1}H^k\cdots A^kH^kA^r$. In the second season, if $r=1$, the pattern is $AHAH^{k-1}A^k\cdots H^kA^{k-1}$; otherwise, the pattern is $H^2AH^{k-1}A^k\cdots H^kA^kH^{r-2}$. The patterns in these two seasons will not create more than $k$-consecutive home/away games.
\end{proof}

\begin{theorem}\label{feasibility}
The construction is feasible for any $k\geq 2$.
\end{theorem}
\begin{proof}
Recall that $\Gamma_i$ denotes the $m$ games of the schedule on the $i$-th day.

First, we prove that the schedule is a complete double round-robin. It is well-known that the rotation scheme in the construction can decompose a complete graph $G$ into $n-1$ edge-disjoint perfect matchings, i.e., $G=\Gamma_1\cup \Gamma_2\cup\cdots \Gamma_{n-2}\cup \Gamma_{n-1}$ if we take $\Gamma_{i'}$ as an undirected matching by taking each game $t_i\rightarrow t_j$ on that day as an edge $t_it_j$. Analogously, we can take $\Gamma_{i'}$ (resp., $\overline{\Gamma_{i'}}$) as a directed matching by taking each game $t_i\rightarrow t_j$ on that day as a directed edge $t_it_j$ (resp., $t_jt_i$), and then $\Gamma_1\cup \Gamma_2\cup\cdots \Gamma_{n-2}\cup \Gamma_{n-1}\cup \overline{\Gamma_1\cup \Gamma_2\cup\cdots \Gamma_{n-2}\cup \Gamma_{n-1}}$ forms a bi-directed complete graph. Note that the bi-directed complete graph represents the $n(n-1)$ games. Hence, all games are contained exactly once in the $2(n-1)$ days $\Gamma_1\cup \Gamma_2\cup\cdots \Gamma_{n-2}\cup \Gamma_{n-1}\cup \overline{\Gamma_1\cup \Gamma_2\cup\cdots \Gamma_{n-2}\cup \Gamma_{n-1}}$. The schedule is a complete double round-robin.

Then, we prove that the schedule satisfies the no-repeat constraint. For each pair of teams, the two games between them are arranged in $\Gamma_i$ and $\overline{\Gamma_i}$ with some $i$, respectively. It is easy to see that the schedule $\Gamma_1\cdot \Gamma_2\cdots \Gamma_{n-2}\cdot \Gamma_{n-1}\cdot \overline{\Gamma_{n-2}\cdot \Gamma_{n-1}\cdot \Gamma_1\cdot \Gamma_2\cdots \Gamma_{n-3}}$ satisfies the constraint.

Last, for the bounded-by-$k$ constraint, it follows directly from Lemma~\ref{feasibility1}.
\end{proof}

Next, we are ready to analyze the approximation quality of our algorithm in Section~\ref{mainalgorithm}.

\section{Analyzing the Approximation Quality}\label{analysis}
In this section, we introduce a novel framework that differs from the one presented in \cite{westphal2014}. Our refined analysis within this framework yields an improved 5-approximation ratio. Importantly, this ratio holds for a wide range of meaningful values for the parameters $k$ and $n$. Specifically, it applies to any $k$ and $n$ satisfying the conditions $n>k\geq 2$ and $n \geq 4$.

To facilitate the analysis, we use $\OPT$ to denote the total traveling distance of all teams in an optimal solution.
Next, we present several useful bounds.

\subsection{Some Bounds}\label{bounds}
First, we recall some well-known lower bounds for TTP-$k$.

\begin{lemma}[\cite{westphal2014}]\label{lb5}
$nd(T^*)\leq \OPT$.
\end{lemma}

\begin{lemma}[\cite{westphal2014}]\label{lb6}
$\Delta\leq (k/2)\cdot \OPT$.
\end{lemma}

Note that Imahori \emph{et al.} \cite{imahori2010approximation} proved that $\Delta\leq (n^2/4)d(T^*)$. By Lemma~\ref{lb5}, we have the following lemma.

\begin{lemma}[\cite{imahori2010approximation}]\label{lb7}
$\Delta\leq (n/4)\cdot \OPT$.
\end{lemma}

Then, to analyze the weight of our schedule, we define some notations and explore their properties.
Recall that the construction contains $2(n-1)$ days, and each day contains $m=n/2$ games. Let $d_i(j)$ (resp., $e_i(j)$) denote the distance (resp., edge) between two vertices representing the two teams in the $i$-th game (from left to right) on the $j$-th day, where $1\leq i\leq m$ and $1\leq j\leq 2(n-1)$. For example, in Figure~\ref{fig01}, we have $d_1(1)=d(20,10)$ and $d_2(1)=d(1,19)$.

\begin{lemma}\label{b1}
$\sum_{i=1}^{m}\sum_{j=1}^{2(n-1)}d_i(j)=\Delta$.
\end{lemma}
\begin{proof}
Since the schedule is a double round-robin, it contains exactly two games between any pair of two teams. Recall that $\Delta$ is twice the total distance of all edges in $G$. So, the lemma holds.
\end{proof}

\begin{lemma}\label{b2}
$\sum_{j=1}^{2(n-1)}d_1(j)=2s(n)$.
\end{lemma}
\begin{proof}
It follows directly from that the most left game always involves team $t_n$, and $t_n$ meets every other team twice in the double round-robin schedule.
\end{proof}

\begin{lemma}\label{b3}
For any $1\leq i\leq m$ and $1\leq j\leq 2(n-1)$, we have $ \EE{d_i(j)}=\frac{1}{2(n-1)}\sum_{j=1}^{2(n-1)}d_i(j)$.
\end{lemma}
\begin{proof}
The construction is based on the rotation scheme, which makes sure that in the construction the positions of the $n-1$ teams $\{t_1,t_2,\dots,t_{n-1}\}$ always correspond to the same cycle $T'$ on different days. Recall that we label $T'$ in the order of $1,2,\dots,n-1$ uniformly at random from the $n-1$ cases and $d_i(j)$ is the distance of edge $e_i(j)$.

\textbf{Case~1: $i=1$.} It is easy to see that $e_i(j)$ corresponds to an edge incident to $n$. There are $n-1$ edges with a total weight of $s(n)$. The probability of each of them being $e_i(j)$ is $\frac{1}{n-1}$. Hence, by Lemma~\ref{b2}, we have $\EE{d_i(j)}=\frac{1}{n-1}\cdot s(n)=\frac{1}{2(n-1)}\sum_{j=1}^{2(n-1)}d_i(j)$.

\textbf{Case~2: $i>1$.} By the construction, we know that $e_i(j)$ corresponds to an edge with the labels of its two vertices differing by $i-1$ on $T'$. Since $T'$ contains $n-1$ (odd) vertices, there are $n-1$ such edges with a total weight of $\frac{1}{2}\sum_{j=1}^{2(n-1)}d_i(j)$. The probability of each of them being $e_i(j)$ is also $\frac{1}{n-1}$. Hence, we have $\EE{d_i(j)}=\frac{1}{2(n-1)}\sum_{j=1}^{2(n-1)}d_i(j)$.

So, the lemma holds.
\end{proof}

Note that the part $\sum_{j=1}^{2(n-1)}d_i(j)$ in Lemma~\ref{b3} is a constant for any $1\leq i\leq m$ (here the constant means a value that stays unchanged for the $n-1$ cases of labeling $1,2,\dots,n-1$). For $i=1$, we have $\sum_{j=1}^{2(n-1)}d_i(j)=2s(n)$ by Lemma~\ref{b2}, which is a constant obviously. For $i>1$, by the labels of teams in the schedule, $\sum_{j=1}^{2(n-1)}d_i(j)$ is twice the total distance of all edges where the labels of vertices of each edge differ by $i-1$ on the cycle $T'$, which is also a constant.

\begin{lemma}\label{lb1}
For any $1\leq j\leq 2(n-1)$, we have $\EE{\sum_{i=1}^{m}d_i(j)}=\frac{\Delta}{2(n-1)}$.
\end{lemma}
\begin{proof}
Note that $\EEE{\sum_{i=1}^{m}d_i(j)}=\sum_{i=1}^{m}\EE{d_i(j)}$. By Lemmas~\ref{b1} and \ref{b3}, we have
\[
\sum_{i=1}^{m}\EE{d_i(j)}=\frac{1}{2(n-1)}\sum_{i=1}^{m}\sum_{j=1}^{2(n-1)}d_i(j)=\frac{\Delta}{2(n-1)}.
\]
\end{proof}

\begin{lemma}\label{lb2}
For any $1\leq i\leq m$ and $1\leq j\leq 2(n-1)$, we have $\EE{d_i(j)}\leq 2\cdot\EE{d_1(j)}$.
\end{lemma}
\begin{proof}
The case $i=1$ obviously holds. We consider $i>1$. By Lemma~\ref{b3}, it is sufficient to prove
\[
\sum_{j=1}^{2(n-1)}d_i(j)\leq 2\sum_{j=1}^{2(n-1)}d_1(j).
\]

In the construction, there are two teams in the $i$-th game, from left to right, on the $j$-th day, including one upper team and one lower team for each $i>1$. Let $a_i(j)$ and $b_i(j)$ be the vertices representing the upper and lower teams, respectively. For example, in Figure~\ref{fig01}, we have $a_2(1)=1$ and $b_2(1)=19$. Then, we have $d_i(j)=d(a_i(j), b_i(j))$. By the triangle inequality, we have $d(a_i(j), b_i(j))\leq d(a_i(j), n)+d(n, b_i(j))$. Hence,
\[
\sum_{j=1}^{2(n-1)}d_i(j)=\sum_{j=1}^{2(n-1)}d(a_i(j), b_i(j))\leq \sum_{j=1}^{2(n-1)}(d(a_i(j), n)+d(n, b_i(j))).
\]

By the rotation scheme of the construction, it is easy to see that each team of $\{t_1,t_2,\dots,t_{n-1}\}$ goes through twice on the upper and lower positions, respectively. Hence, by Lemma~\ref{b2}, we have
\[
\sum_{j=1}^{2(n-1)}(d(a_i(j), n)+d(n, b_i(j)))=\sum_{i=1}^{n-1}(2d(i,n)+2d(n,i))=4s(n)=2\sum_{j=1}^{2(n-1)}d_1(j).
\]
\end{proof}

\begin{lemma}\label{lb3}
For any $1\leq j\leq 2(n-1)$, we have $\EE{d_1(j)}\leq\frac{\Delta}{n(n-1)}$.
\end{lemma}
\begin{proof}
By Lemmas~\ref{b2} and~\ref{b3}, we have
\[
\EE{d_1(j)}=\frac{1}{2(n-1)}\sum_{j=1}^{2(n-1)}d_1(j)=\frac{1}{2(n-1)}\cdot 2s(n)\leq\frac{\Delta}{n(n-1)},
\]
where the inequality follows from $s(n)\leq(1/n)\Delta$ because the vertex $n$ is selected with the total distance of edges incident on it minimized in the algorithm.
\end{proof}

\begin{lemma}\label{lb4}
For any $1\leq j\leq 2(n-1)$, we have $\EE{d_2(j)}=\frac{d(T')}{n-1}$.
\end{lemma}
\begin{proof}
Recall that $\sum_{j=1}^{2(n-1)}d_i(j)$ is twice the total distance of all edges where the labels of vertices of each edge differ by $i-1$ on the cycle $T'$. So, we have $\sum_{j=1}^{2(n-1)}d_2(j)=2d(T')$.
By Lemma~\ref{b3}, we have
\[
\EE{d_2(j)}=\frac{1}{2(n-1)}\sum_{j=1}^{2(n-1)}d_2(j)=\frac{1}{2(n-1)}\cdot 2d(T')=\frac{d(T')}{n-1}.
\]
\end{proof}

\subsection{The Analysis}\label{ouranalysis}
Recall that all teams are at home on the $0$-th and $(2n-1)$-th days. To analyze the expected cost of the schedule, we calculate the total expected cost of all \emph{moves} of all teams in the schedule, where all teams make one move from the game venue on the $i$-th day to the game venue on the $(i+1)$-th day ($0\leq i\leq 2n-2$). Note that if a team plays two home games on two consecutive days, then the cost of the move is 0.
If a team plays two away games on two consecutive days, we call the move an \emph{away-move}, which corresponds to a move between two opponents' venues. If a team plays only one away game on two consecutive days, we call the move a \emph{home-move}, which corresponds to a move from/to its home venue. Then, we only need to consider all teams' away-moves and home-moves.

For the sake of analysis, we make two assumptions that will not decrease the cost of the schedule by the triangle inequality.

\begin{assumption}\label{ass1}
We assume that there is a home day between the last day of the first season and the first day of the second season, where all teams are at home.
\end{assumption}

After Assumption~\ref{ass1}, we can analyze the total expected cost of all moves from the $0$-th day to the home day in the first season and then from the home day to the $(2n-1)$-th day in the second season, separately.

\begin{assumption}\label{ass2}
Supposing team $t_i$ plays with $t_n$ on the $j$-th day, we assume that both $t_i$ and $t_n$ return home after the games on the $(j-1)$-th and the $j$-th days, respectively.
\end{assumption}

Both Assumptions \ref{ass1} and \ref{ass2} will not decrease the cost of the schedule by the triangle inequality. Typically, a team only makes one move on two consecutive days. However, after Assumption~\ref{ass2}, a team may make an extra move. For example, $t_i$ plays an away game with $t_j$ on the $(i'-1)$-th day, and then plays an away game with $t_n$ on the $i'$-th day, where $t_i$ should make one move from the venues of $t_j$ to $t_n$. After the assumption, $t_i$ will make two moves: one is from the venues of $t_j$ to $t_i$ on the $(i'-1)$-th day, and the other is from the venues of $t_i$ to $t_n$ on the $i'$-th day. Both moves of $t_i$ involve its home venue, and hence are home-moves. Moreover, after Assumption~\ref{ass2}, a team makes an away-move only if it plays two consecutive away games in the same block within one season.

Since $\EE{d_i(j)}$, the expected distance of edge $e_i(j)$, is always the same for any $1\leq j\leq 2(n-1)$, we simply let $d_i\coloneqq\EE{d_i(j)}$ and denote the corresponding edge by $e_i$.

Next, we analyze the total expected cost of all moves in the first season.

\textbf{From the $0$-th day to the $1$-st day.} Since all teams are at home on the $0$-th day and there are $m$ games on the $1$-st day, we know that there are $m$ home-moves corresponding to the $m$ edges $e_1,\dots,e_m$, respectively. The total expected cost of all moves is
\begin{align}\label{cost1}
\sum_{i=1}^{m}d_i=\frac{\Delta}{2(n-1)},
\end{align}
where the equation follows from Lemma~\ref{lb1}.

\textbf{From the $i$-th day to the $(i+1)$-th day ($1\leq i\leq n-2$).} There are two cases: $k<n/2$ and $k\geq n/2$.

\textbf{Case~1: $k<n/2$.} Recall that the algorithm selects the best $l$ from $\{1,2,\dots,k\}$, and the construction contains $b=\ceil{\frac{m-l-1}{k}}+2$ blocks, where the widths are $w_1=1$, $w_2=m-1-(b-3)k-l$, $w_i=k$ ($i=3,\dots, b-1$), and $w_b=l$.
So, these $b$ blocks can be presented by
\[
(e_1), (e_{2},\dots, e_{1+w_2}), (e_{2+w_2},\dots, e_{1+w_2+w_3}), \dots, (e_{m-1+l},\dots, e_{m}).
\]

We first consider away-moves. After Assumption~\ref{ass2}, an away-move happens only in the same block. By the construction, it corresponds to an edge of $T'$. Since we label $T'$ uniformly at random from the $n-1$ cases, the expected cost of an away-move is $\frac{d(T')}{n-1}$, which equals to $d_2$ by Lemma~\ref{lb4}. Recall that $\sum_{i=1}^{b}w_i=m$. Typically, there are $w_i-1$ away-moves in the $i$-th block with width $w_i$, and then $\sum_{i=1}^{b}(w_i-1)=m-b$ away-moves in total. One special case is that when $w_2=k$, i.e., $(m-l-1)\bmod k=0$, we further reverse the direction of the most left edge in the $2$-nd block, there will be one less away-move in the $2$-nd block, and then we have $m-b-1$ away-moves in total (we can imagine that the $2$-nd block is divided into two blocks, and hence there are $b+1$ blocks in total). For all $k$ cases of $l\in\{1,2,\dots, k\}$, there is only one case with $w_2=k$. Recall that $b=\Ceil{\frac{m-1-l}{k}}+2$. Hence, for all these $k$ cases, the total expected cost of all away-moves are
\begin{align*}
&\sum_{l=1}^{k}\lrA{m-\Ceil{\frac{m-1-l}{k}}-2}d_2-d_2\\
&\leq \sum_{l=1}^{k}\lrA{m-\Ceil{\frac{m-1-l}{k}}-2}d_2\\
&=(l_0-1)\lrA{m-\frac{m-1-l_0}{k}-3}d_2+(k-l_0+1)\lrA{m-\frac{m-1-l_0}{k}-2}d_2\\
&=(k-1)(m-2)d_2,
\end{align*}
where $l_0$ in the second equality is the number satisfying $(m-l_0-1)\bmod k=0$ and $l_0\in\{1,2,\dots,k\}$.

Then, we consider home-moves. Consider an arbitrary block $(e_i,\dots, e_{j})$, where we call the edges $e_i$ and $e_j$ \emph{boundary edges}. Note that the boundary edges may be the same. After Assumption~\ref{ass2}, we have that there are two home-moves corresponding to the \emph{boundary edges}. For example, in Figure~\ref{fig01}, we consider the $1$-st block $(e_1)$ from the $1$-st day to the $2$-nd day. We can see that $t_{20}$ takes a home-move from the venues of $t_{10}$ to $t_{20}$ on the $1$-st day and a home-move from the venues of $t_{20}$ to $t_1$ on the $2$-nd day. Both home-moves correspond to the same boundary edge $e_1$. Note that for the special case with $w_2=k$, we imagine that the $2$-nd block $(e_2,e_3,\dots,e_{k+1})$ is further divided into two blocks $(e_2)$ and $(e_3,\dots,e_{k+1})$, and then we will have two more home-moves corresponding to the edges $e_2$ and $e_3$, respectively. For all $k$ cases of $l\in\{1,2,\dots,k\}$, there must be two cases such that the edge $e_i$ ($3\leq i\leq m-1$) is a boundary edge in some block, one case such that the edge $e_2$ (resp., $e_m$) is the only boundary edge in the $2$-nd block (resp., the last block). Hence, for all these $k$ cases, the total expected cost of all home-moves are
\begin{align*}
&2kd_1+(k+1)d_2+2\sum_{i=3}^{m-1}d_i+(k+1)d_m+(d_2+d_3)\\
&=2kd_1+kd_2+kd_m+2\sum_{i=1}^{m}d_i+d_3-2d_1-d_m\\
&\leq 4kd_1+kd_2+2\sum_{i=1}^{m}d_i\\
&=4kd_1+kd_2+\frac{\Delta}{n-1},
\end{align*}
where the inequality follows from $d_m\leq 2d_1$ and $d_3\leq 2d_1$ by Lemma~\ref{lb2}, and the last equality follows from Lemma~\ref{lb1}.

For all $k$ cases of $l\in\{1,2,\dots,k\}$, there must be one case such that the total expected cost of all moves is bounded by
\begin{align}\label{cost2.1}
4d_1+d_2+\frac{\Delta}{k(n-1)}+\frac{(k-1)(m-2)d_2}{k}.
\end{align}

\textbf{Case~2: $k\geq n/2$.} There are only two blocks: $(e_1)$ and $(e_2,\dots,e_m)$. There are $m-2$ away-moves with an expected cost of $(m-2)d_2$, and four home-moves with an expected cost of $2d_1+d_2+d_m$. Therefore, the total expected cost of all moves is
\begin{align}\label{cost2.2}
2d_1+d_2+d_m+(m-2)d_2=2d_1+(m-1)d_2+d_m\leq 4d_1+(m-1)d_2,
\end{align}
where the inequality follows from $d_m\leq 2d_1$ by Lemma~\ref{lb2}.

\textbf{From the $(n-1)$-th day to the home day.} Since there are $m$ games on the $(n-1)$-th day and all teams are at home on the home day, we know that there are $m$ home-moves corresponding to the edges $e_1,\dots,e_m$, respectively. The total expected cost of all moves is
\begin{align}\label{cost3}
\sum_{i=1}^{m}d_i=\frac{\Delta}{2(n-1)},
\end{align}
where it follows from Lemma~\ref{lb1}.

We can analyze the total expected cost of all moves in the first season.

\medskip
\noindent
\textbf{Case~1: $k\leq n/2$.}
For the best choice of $l$, by (\ref{cost1}), (\ref{cost2.1}), and (\ref{cost3}), the expected cost of moves in the first season is
\begin{align*}
&\frac{\Delta}{2(n-1)}+(n-2)\lrA{4d_1+d_2+\frac{\Delta}{k(n-1)}+\frac{(k-1)(m-2)d_2}{k}}+\frac{\Delta}{2(n-1)}\\
&\leq\frac{\Delta}{n-1}+4(n-2)d_1+(n-1)d_2+(1/k)\Delta+(1-1/k)(m-2)(n-1)d_2\\
&\leq \frac{\Delta}{n-1}+\frac{4(n-2)\Delta}{n(n-1)}+d(T')+(1/k)\Delta+(1-1/k)(m-2)d(T')\\
&= \frac{(5n-8)\Delta}{n(n-1)}+d(T')+(1/k)\Delta+(1-1/k)(m-2)d(T')\\
&\leq (5/n)\Delta+(1/k)\Delta+(1-1/k)md(T),
\end{align*}
where the second inequality follows from $d_1\leq\frac{\Delta}{n(n-1)}$ and $d_2=\frac{d(T')}{n-1}$ by Lemmas~\ref{lb3} and \ref{lb4}, and the last inequality follows from $d(T')\leq 2(1-1/k)d(T')$ and $d(T')\leq d(T)$ by the triangle inequality.

\medskip
\noindent
\textbf{Case~2: $k\geq n/2$.}
For the best choice of $l$, by (\ref{cost1}), (\ref{cost2.2}), and (\ref{cost3}), the expected cost of moves in the first season is
\begin{align*}
&\frac{\Delta}{2(n-1)}+(n-2)\lrA{4d_1+(m-1)d_2}+\frac{\Delta}{2(n-1)}\\
&\leq\frac{\Delta}{n-1}+4(n-2)d_1+(n-1)(m-1)d_2\\
&\leq \frac{\Delta}{n-1}+\frac{4(n-2)\Delta}{n(n-1)}+(m-1)d(T')\\
&= \frac{(5n-8)\Delta}{n(n-1)}+(1-1/m)md(T')\\
&\leq (5/n)\Delta+(1-1/k)md(T),
\end{align*}
where the second inequality follows from $d_1\leq\frac{\Delta}{n(n-1)}$ and $d_2=\frac{d(T')}{n-1}$ by Lemmas~\ref{lb3} and \ref{lb4}, and the last inequality follows from $m=n/2\leq k$ and $d(T')\leq d(T)$ by the triangle inequality.

Note that based on Assumptions~\ref{ass1} and \ref{ass2}, we can easily get that the total expected cost of all moves in the second season is the same as that in the first season. So, the total expected cost is at most $(10/n)\Delta+(2/k)\Delta+(1-1/k)nd(T)$ for $k<n/2$ and $(10/n)\Delta+(1-1/k)nd(T)$ for $k\geq n/2$.

Since we label the vertices on $T'$ uniformly at random from the $n-1$ cases, to derandomize the algorithm, we can simply enumerate all these cases and choose the best one. Consider the running time. In each case, there are $k$ choices of $l$ for $k<n/2$ and one choice of $l$ for $k\geq n/2$. Moreover, for each fixed $l$, the construction takes $O(n^2)$ time. Hence, the running time of the deterministic algorithm is $O(n^3k)$ for $k<n/2$ and $O(n^3)$ for $n\geq k/2$. We can get the following theorem.

\begin{theorem}\label{maincost}
For TTP-$k$ with $k\geq 2$, there is a polynomial-time algorithm that can generate a solution with a weight of at most $(10/n)\Delta+(2/k)\Delta+(1-1/k)nd(T)$, where $T$ is a given Hamiltonian cycle of graph $G$. If $k\geq n/2$, the upper bound can be improved to $(10/n)\Delta+(1-1/k)nd(T)$.
\end{theorem}

Note that Yamaguchi \emph{et al.}~\cite{yamaguchi2009improved} proposed an algorithm that can generate a solution with a weight of at most $O(1/n)\Delta+(2/k)\Delta+(1-1/k)nd(T)$ for TTP-$k$ with $k\geq 3$. Theorem~\ref{maincost} implies that the construction in this paper has a better upper bound than that in~\cite{yamaguchi2009improved}.
Their algorithm is more complicated, and it is not even easy to compute an upper bound $(c/n)\Delta+(2/k)\Delta+(1-1/k)d(T)$ with some small constant $c$.
Since we have $\Delta=O(k)\cdot\OPT$ by Lemma~\ref{lb6}, their result only implies an $O(1)$-approximation for TTP-$k$ with $k=\Theta(n)$.
It is worth noting that their schedule satisfies an additional mirrored constraint, i.e., the second season is directly arranged by reversing the home venues in the first season. Therefore, we suspect that the constant $c$ in their algorithm is much larger than $10$, even by a tighter analysis.

\subsection{The Approximation Ratio}\label{ourratio}
Recall that the construction is based on a given Hamiltonian cycle $T$ of graph $G$. We simply use the well-known $3/2$-approximation algorithm, which takes $O(n^3)$ time. Then, we have the following lemma.

\begin{lemma}[\cite{christofides1976worst,serdyukov1978some}]\label{tsp}
$d(T)\leq (3/2)d(T^*)$.
\end{lemma}

\begin{theorem}
For TTP-$k$ with $k\geq 2$, there is a polynomial-time $5$-approximation algorithm. If $k\geq n/2$, the approximation ratio can be improved to $4$.
\end{theorem}
\begin{proof}
First, we consider $k<n/2$. By Theorem~\ref{maincost}, the weight is bounded by $(10/n)\Delta+(2/k)\Delta+(1-1/k)d(T)$. We can get that
\begin{align*}
&(10/n)\Delta+(2/k)\Delta+(1-1/k)nd(T)\\
&\leq (10/n)\Delta+(2/k)\Delta+(3/2)(1-1/k)nd(T^*)\\
&\leq (5k/n)\cdot \OPT+\OPT+(3/2)(1-1/k)\cdot \OPT\\
&\leq (5/2)\cdot\OPT+\OPT+(3/2)\cdot\OPT\\
&=5\cdot\OPT,
\end{align*}
where the first inequality follows from $d(T)\leq (3/2)d(T^*)$ by Lemma~\ref{tsp}, the second inequality follows from $nd(T^*)\leq\OPT$ and $\Delta\leq (k/2)\cdot \OPT$ by Lemmas~\ref{lb5} and \ref{lb6}, and the last inequality follows from $k<n/2$.

Then, we consider $k\geq n/2$. By Theorem~\ref{maincost}, the weight is bounded by
\begin{align*}
&(10/n)\Delta+(1-1/k)nd(T)\\
&\leq (10/n)\Delta+(3/2)nd(T^*)\\
&\leq (5/2)\cdot \OPT+(3/2)\cdot \OPT\\
&=4\cdot \OPT,
\end{align*}
where the first inequality follows from $d(T)\leq (3/2)d(T^*)$ by Lemma~\ref{tsp}, and the second inequality follows from $nd(T^*)\leq\OPT$ and $\Delta\leq (n/4)\cdot \OPT$ by Lemmas \ref{lb5} and \ref{lb7}.
\end{proof}

\section{Conclusion}\label{conclusion}
In this paper, we present a 5-approximation algorithm for TTP-$k$, which not only addresses previous flaws but also significantly improves previous results.
In our algorithm, we simply use the simple 3/2-approximation algorithm for TSP.
Recently, the TSP ratio was slightly improved to $3/2-\varepsilon$, where the improvement $\varepsilon$ is about $10^{-36}$ \cite{DBLP:conf/stoc/KarlinKG21}. By utilizing the improved approximation algorithm for TSP, we may also be able to slightly improve the result.
For the case $k=n-1$, our ratio is worse than the current-best ratio $2.75$ in~\cite{imahori2010approximation}.
It would be interesting to know whether our technique in this paper can also be used to improve the ratio for the case $k=n-1$.

\section*{Acknowledgments}
The work is supported by the National Natural Science Foundation of China, under grant 61972070.

\section*{Declaration of competing interest}
The authors declare that they have no known competing financial interests or personal relationships that could have appeared to influence the work reported in this paper.

\bibliographystyle{plain}
\bibliography{main}

\newpage
\appendix

\section{The Flaw in the Previous Construction Algorithm}\label{previousconstruction}
Westphal and Noparlik's construction~\cite{westphal2014} only makes sure that the direction of the most left edge (incident on team $t_n$) changes every $k$-th day in the first season. Note that in our construction we deal with one more case: if $r\leq k<n-1$ we further reverse the most left edge on the 1-st day.

We show that team $t_n$ may play more than $k$-consecutive home/away games in Westphal and Noparlik's construction, thus leading to an infeasible schedule for TTP-$k$.

\begin{lemma}\label{infeasibility}
In Westphal and Noparlik's construction, team $t_n$ plays more than $k$-consecutive home games if and only if $r\leq k<n-1$.
\end{lemma}
\begin{proof}
By a similar argument in the proof of Lemma~\ref{feasibility1}, we consider the following two cases. Note that the direction of the most left edge changes every $k$-th day.

\textbf{Case~1: $r>k$.} The home/away pattern of team $t_n$ in the first season is $A^kH^k\cdots A^kH^kA^kH^{r-k}$. In the second season, if $r-k=1$, the pattern is $HAH^kA^k\cdots H^kA^kH^{k-1}$; otherwise, the pattern is $A^2H^kA^k\cdots H^kA^kH^kA^{r-k-2}$. In both cases, the patterns in these two seasons can be combined well without creating more than $k$-consecutive home/away games.

\textbf{Case~2: $r\leq k$.} Similarly, we have $n-1>2k$. The home/away pattern of team $t_n$ in the first season is $A^kH^k\cdots A^kH^kA^r$. In the second season, if $r=1$, the pattern is $AHH^kA^k\cdots H^kA^{k-1}$; otherwise, the pattern is $H^2H^kA^k\cdots H^kA^kH^{r-2}$. In both cases, the pattern in the second season can create more than $k$-consecutive home games.
\end{proof}

Lemma~\ref{infeasibility} shows that when $r\leq k<n-1$, the construction in~\cite{westphal2014} is infeasible. Note that if $k=O(1)$, the construction is infeasible for at least half of the instances. Their construction was also considered for $k=2$ in \cite{thielen2012approximation}, where there are two cases: even $n/2$ and odd $n/2$. But, it fails for the case of odd $n/2$ with the same reason. Some experimental algorithms \cite{thielen2012approximation,westphal2014,goerigk2016combined} take the output of their constructions as an initial solution, and then optimize it using some heuristic methods. However, we are unaware whether the experimental results are correct due to this flaw.

\section{The Flaw in the Previous Analysis}\label{previousanalysis}
Westphal and Noparlik proved an approximation ratio of $2+2k/n+k/(n-1)+3/n+3/(2\cdot k)$~\cite{westphal2014}. For any constant $k>10$, the approximation ratio achieves $2+3/(2\cdot k)+O(1/n)$, which is even better than the approximation ratio $(5k-7)/(2k)+O(1/n)$ in~\cite{yamaguchi2009improved}. Note that the latter approximation ratio is obtained using some stronger lower bounds than the former one. Indeed, there is a flaw in Westphal and Noparlik's analysis. In their analysis, the cost of the schedule was divided into several parts (note that our analysis is based on a different framework).

\begin{lemma}[\cite{westphal2014}]\label{panalysis}
Given a Hamiltonian cycle $T$ of graph $G$, the cost of the schedule is bounded by $C_h+C_a+C_s+C_l+C_r+C_o$, where
\begin{itemize}
\item $C_h\leq (2/n)\Delta$;
\item $C_a\leq (2/n)\Delta$;
\item $C_s\leq \frac{2(\Delta-2s(n))}{n-1}$;
\item $C_l\leq 2d(T)$;
\item $C_r=2\cdot\lrA{\sum_{i=1}^{n/2}d(i,i+n/2-1)+\sum_{i=n/2+1}^{n}d(i,i-n/2)}$;
\item $C_o\leq \frac{2\Delta+(n-2)d(T)}{k}$.
\end{itemize}
\end{lemma}

To see the flaw, we give a counterexample. We consider TTP-$k$ with $k=\Theta(n)$ and the complete graph $G$ where the length of each edge is a unit.
We can get $s(n)=n-1$ and $\Delta=n(n-1)$. Note that the length of any Hamiltonian cycle of graph $G$ is $n$. Hence, we have $d(T)=n$. Therefore, by Lemma~\ref{panalysis}, we have
\begin{align*}
&C_h+C_a+C_s+C_l+C_r+C_o\\
&\leq (2/n)\Delta+(2/n)\Delta+\frac{2(\Delta-2s(n))}{n-1}+2d(T)\\
&\quad\ +2\cdot\lrA{\sum_{i=1}^{n/2}d(i,i+n/2-1)+\sum_{i=n/2+1}^{n}d(i,i-n/2)}+\frac{2\Delta+(n-2)d(T)}{k}\\
&=2(n-1)+2(n-1)+2(n-2)+2n+2n+\frac{2n(n-1)+n(n-2)}{k}\\
&=\Theta(n).
\end{align*}
For this example, the analyzed cost of the schedule is bounded by $\Theta(n)$. However, it is easy to see that in an optimal schedule, the traveling distance of each team is $\Theta(n)$, and then the total traveling distance of all teams is $\Theta(n^2)$, a contradiction.

The flaw is due to the analysis of the cost $C_o$ and the cost of this part is also the main cost of the schedule. Note that a simple refined analysis may lead us to get
\[
C_o\leq (2/k)\Delta+(n-2)d(T),
\]
with an approximation ratio $7/2+2k/n+k/(n-1)+3/n$ instead of $2+2k/n+k/(n-1)+3/n+3/(2\cdot k)$, i.e., the term $3/(2\cdot k)$ in their claimed approximation ratio should be $3/2$. When $k=\Theta(n)$, the approximation ratio is bounded by $6.5+O(1/n)$ instead of the claimed $5+O(1/n)$. For any $n>k\geq 4$ and $n\geq 6$, the approximation ratio is bounded by $6.667$ instead of the claimed $5.875$.
\end{document}